\def\year{2018}\relax
\documentclass[letterpaper]{article}
\usepackage{aaai18}
\usepackage{times}
\usepackage{helvet}
\usepackage{courier}
\usepackage{url}
\usepackage{graphicx}
\usepackage{hyperref}
\usepackage{paralist}
\usepackage{bm}
\usepackage[utf8]{inputenc}
\usepackage{url}
\usepackage{amsmath}
\usepackage{amsfonts}
\usepackage{amsthm}
\usepackage{soul} 
\usepackage{amssymb}
\usepackage{todonotes}
\usepackage[linesnumbered,ruled, vlined]{algorithm2e}

\newtheorem{lemma}{Lemma}
\newtheorem{theorem}{Theorem}

\newtheorem{remark}{Remark}
\newtheorem{prop}{Proposition}
\newtheorem{observation}{Observation}
\newtheorem{conjecture}{Conjecture}

\newtheorem{reduction}{Reduction}

\newcommand{\mectime}{\textsc{mec}\xspace}
\newcommand{\attr}{\textit{Attr}}
\DeclareMathOperator*{\argmax}{argmax}

\newcommand{\best}{\textit{best}}
\SetKwProg{Function}{function}{}{}

\newcommand{\mcount}{\textit{count}}
\newcommand{\nill}{\textit{null}}

\newcommand{\ls}{\langle}
\newcommand{\rs}{\rangle}

\newcommand{\poly}{\mathit{poly}}

\newcommand{\Reach}{\mathit{Reach}}
\newcommand{\Out}{\mathit{Out}}
\newcommand{\In}{\mathit{In}}
\newcommand{\Seq}{\mathit{Seq}}
\newcommand{\GG}{\Gamma}
\newcommand{\C}{\mathcal{C}}
\newcommand{\D}{\mathcal{D}}

\newcommand{\T}{\mathcal{T}}

\title{Algorithms and Conditional Lower Bounds for Planning Problems}
\author{Krishnendu Chatterjee \\
		IST Austria\\	
		\texttt{\scriptsize krish.chat@ist.ac.at}
		\And
		Wolfgang Dvo\v{r}\'{a}k\\ 
		TU Wien\\
		\texttt{\scriptsize dvorak@dbai.tuwien.ac.at}
		\And
		Monika Henzinger\\
		University of Vienna\\
		\texttt{\scriptsize monika.henzinger@univie.ac.at}
		\And
		Alexander Svozil\\
		University of Vienna\\
		\texttt{\scriptsize alexander.svozil@univie.ac.at}
}
\begin{document}

\maketitle
\begin{abstract}
We consider planning problems for graphs, Markov decision processes (MDPs),
and games on graphs. 
While graphs represent the most basic planning model, MDPs represent
interaction with nature and games on graphs represent interaction with 
an adversarial environment.
We consider two planning problems where there are $k$ different target sets,
and the problems are as follows: 
(a)~the coverage problem asks whether there is a 
plan for each individual target set; and 
(b)~the sequential target reachability problem asks whether the targets can be 
reached in sequence.
For the coverage problem, we present a linear-time algorithm for graphs,
and quadratic conditional lower bound for MDPs and games on graphs.
For the sequential target problem, we present a linear-time algorithm for graphs,
a sub-quadratic algorithm for MDPs, and a quadratic conditional lower bound for 
games on graphs.
Our results with conditional lower bounds establish 
(i)~model-separation results showing that for the coverage problem MDPs 
and games on graphs are harder than graphs, and for the sequential reachability 
problem games on graphs are harder than MDPs and graphs;
and 
(ii)~objective-separation results showing that for MDPs the coverage problem is harder than the sequential target problem.
\end{abstract}

\section{Introduction}

\noindent{\bf Planning models.} 
One of the basic and fundamental algorithmic problems in 
artificial intelligence is the {\em planning problem}~\cite{LaValle,AIBook}.
The classical models in planning are as follows:
\begin{itemize}
\item {\em Graphs.} The most basic planning problems are graph search 
problems~\cite{LaValle,AIBook}.

\item {\em MDPs.} In the presence of interaction with nature, the graph model is 
extended with probabilities or stochastic transitions, which gives rise to 
Markov decision processes (MDPs)~\cite{Howard,Puterman,FV97,PT87}.

\item {\em Games on graphs.} In the presence of interaction with an adversarial 
environment, the graph model is extended to AND-OR graphs 
(or games on graphs)~\cite{JACM85,hansen98andor}.
\end{itemize}
Thus graphs, MDPs, and games on graphs are the fundamental models for planning.

\smallskip\noindent{\bf Planning objectives.} 
The planning objective represents the goal that the planner seeks to achieve. 
Some basic planning objectives are as follows:
\begin{itemize}
\item {\em Basic target reachability.} Given a set $T$ of target vertices 
the planning objective is to reach some target vertex from the starting position.

\item {\em Coverage Objective.} 
In case of coverage there are $k$ different target sets, namely, 
$T_1,T_2,\ldots,T_k$, and the planning objective asks whether for each $1\leq i \leq k$ 
the basic target reachability with target set $T_i$ can be achieved. 
The coverage models the following scenarios: Consider that 
there is a robot or a patroller, and there are $k$ different target locations,
and if an event or an attack happens in one of the target locations, then that location 
must be reached.
However, the location of the event or the attack is not known in advance and the planner
must be prepared that the target set could be any of the $k$ target sets.

\item {\em Sequential target reachability.} 
In case of sequential targets there are $k$ different target sets, namely, 
$T_1,T_2,\ldots,T_k$, and the planning objective asks first to reach $T_1$,
then $T_2$, then $T_3$ and so on. This represents the scenario 
that there is a sequence of tasks that the planner must achieve.

\end{itemize}
The above are the most natural planning objectives and have been studied in 
the literature, e.g., in robot planning~\cite{KGFP09,kaelbling1998planning,choset2005principles}.

{\small
	\begin{table*}
		\begin{center}
			\begin{tabular}{|c|c|c|c|}
				\hline
				\ Obj/Model \ & Graphs & MDPs & Games on graphs \\
				\hline
				\hline
				\ Basic target \ & $O(m)$ & $O(m \cdot n^{2/3})$ & $O(m)$ \\
				\hline
				\ Coverage objective\ & $O(m + \sum_{i=1}^k |T_i|)$ &  $O(m \cdot n^{2/3}+ k\cdot m)$ & $O(k\cdot m)$ \\
				& 	& $\bm{\tilde{\Omega}(k \cdot
					m)}$ [Thm.~\ref{thm:cover:mdpsparse},\ref{thm:cover:mdpsdense}] &
				$\bm{\tilde{\Omega}(k\cdot
					m)}$ [Thm.~\ref{thm:cover:gamessparse},\ref{thm:cover:gamesdense}]\\
				\hline
				\ Sequential target \ & $\bm{O(m + \sum_{i=1}^k
					|T_i|)}$  & $\bm{O(m \cdot n^{2/3}+
					\sum_{i=1}^k |T_i|)}$ & $O(k\cdot m)$ \\
				& [Thm.~\ref{thm:seq:graphs:upper}]	&[Thm.~\ref{thm:seq:mdp:upper}]  & $\bm{\tilde{\Omega}(k\cdot m)}$
				[Thm.~\ref{thm:seq:games:lower:sparse},\ref{thm:seq:games:lower:dense}]\\
				\hline
			\end{tabular}
			\caption{Algorithmic bounds where $n$ and $m$ are the number of vertices and edges of the underlying model, 
				and $k$ denotes the number of different target sets.
				The $\tilde{\Omega}(\cdot)$ bounds are conditional lower bounds (CLBs) under the BMM conjecture
				and SETH. They establish that polynomial improvements over the given bound are not possible, however,
				polylogarithmic improvements are not excluded. 
				Note that CLBs are quadratic for $k=\Theta(n)$. The new results are highlighted in boldface.}\label{tab:complexity}
		\end{center}
		\vspace{-1em}
	\end{table*}
}

\smallskip\noindent{\bf Planning questions.}
For the above planning objectives the basic planning questions are as follows:
(a)~for graphs, the question is whether there exists a plan (or a path) such 
that the planning objective is satisfied;
(b)~for MDPs, the basic question is whether there exists a policy such that the 
planning objective is satisfied almost-surely (i.e., with probability~1);
and (c)~for games on graphs, the basic question is whether there exists a policy 
that achieves the objective irrespective of the choices of the adversary.
The almost-sure satisfaction for MDPs is also known as the strong cyclic planning 
in the planning literature~\cite{CPRT03}, and games on graphs question represent 
planning in the presence of a worst-case adversary~\cite{JACM85,hansen98andor} (aka 
adversarial planning, strong planning~\cite{MBKS14}, or conformant/contingent
planning~\cite{bonet2000planning,HOFFMANN2005contingent,palacios2007conformant}).

\smallskip\noindent{\bf Algorithmic study.}
In this work, we consider the algorithmic study of the planning questions for 
the natural planning objectives for graphs, MDPs, and games on graphs.
For all the above questions, polynomial-time algorithms exist.
When polynomial-time algorithms exist, proving an unconditional lower bound is 
extremely rare. 
A new approach in complexity theory aims to establish conditional lower bound (CLB) 
results based on some well-known conjecture.
Two standard conjectures for CLBs are as follows: The 
(a)~{\em Boolean matrix multiplication (BMM)} conjecture which states that there is 
no sub-cubic combinatorial algorithm for boolean matrix multiplication; and 
the (b)~{\em Strong exponential-time hypothesis (SETH)} which states that there is 
no sub-exponential time algorithm for the SAT problem.
Many CLBs have been established based on the above conjectures, e.g., 
for dynamic graph algorithms, string matching~\cite{abboud2014popular,bringmann2015quadratic}.

\smallskip\noindent{\bf Previous results and our contributions.}
We denote by $n$ and $m$ the number of vertices and edges of the underlying model, 
and $k$ denotes the number of different target sets.
For the basic target reachability problem, while the graphs and games on graphs problem 
can be solved in linear time~\cite{beeri1980membership,immerman1981number}, 
the current best-known bound for MDPs is 
$O(m \cdot n^{2/3})$~\cite{henzinger2014efficient,chatterjee2016separation}.
For the coverage and sequential target reachability, an $O(k \cdot m)$ upper 
bound follows for graphs and games on graphs, and an $O( m \cdot n^{2/3} + k \cdot m)$ upper 
bound follows for MDPs.
Our contributions are as follows:
\begin{compactenum}
\item {\em Coverage problem:} 
First, we present an $O(m + \sum_{i=1}^k |T_i|)$ time algorithm for graphs; 
second, we present an $\Omega(k\cdot m)$ lower bound for MDPs and games on graphs, 
both under the BMM conjecture and the SETH.
Note that for graphs our upper bound is linear time, however, if each $|T_i|$ is 
constant and $k=\theta(n)$, for MDPs and games on graphs the CLB is quadratic.

\item {\em Sequential target problem:}
First, we present an $O(m + \sum_{i=1}^k |T_i|)$ time algorithm for graphs;
second, we present  an $O( m \cdot n^{2/3} +  \sum_{i=1}^k |T_i|)$ time algorithm for MDPs;
and third, we present an $\Omega(k\cdot m)$ lower bound for games on graphs, 
both under the BMM conjecture and the SETH.

\end{compactenum}
The summary of the results is presented in Table~\ref{tab:complexity}. Our 
most interesting results are the conditional lower bounds for MDPs and game graphs
for the coverage problem, the sub-quadratic algorithm for MDPs with sequential
targets, and the conditional lower bound for game graphs with sequential targets.

\smallskip\noindent{\bf Practical Significance.}
The sequential reachability and coverage problems we consider are the tasks
defined in~\cite[Section II. PROBLEM FORMULATION, 3) System
Specification]{KGFP09}, where the problems have been studied for games on graphs
(Section IV. DISCRETE SYNTHESIS) and mentioned as future work for MDPs (Section
I. INTRODUCTION, A. Related Work). The
applications of these problems have been demonstrated in robotics applications.
We present a complete algorithmic picture for games on graphs and MDPs, settling
open questions related to games and future work mentioned in~\cite{KGFP09}.

\smallskip\noindent{\bf Theoretical Significance.}
Our results present a very interesting algorithmic picture for the natural planning 
questions in the fundamental models.
\begin{compactenum}
\item 
First, we establish results showing that some models are harder than others.
More precisely,
\begin{compactitem}
\item for the basic target problem, the MDP model seems harder than graphs/games on 
graphs (linear-time algorithm for graphs and games on graphs, and no such
algorithms are known for MDPs);
\item for the coverage problem, MDPs, and games on graphs are harder than graphs
(linear-time algorithm for graphs and quadratic CLBs for MDPs and games on graphs);
\item for the sequential target problem, games on graphs are harder than MDPs and graphs 
(linear-time upper bound for graphs and sub-quadratic upper bound for MDPs, whereas 
quadratic CLB for games on graphs).
\end{compactitem}
In summary, we establish model-separation results with CLBs: 
For the coverage problem, MDPs and games on graphs are algorithmically harder than graphs; and 
for the sequential target problem, games on graphs are algorithmically harder than MDPs and graphs.

\item Second, we also establish objective-separation results. 
For the model of MDPs consider the different objectives:
Both for basic target and sequential target reachability the upper bound is 
sub-quadratic and in contrast to the coverage problem we establish a 
quadratic CLB. 

\end{compactenum}

\smallskip\noindent\emph{Discussion related to other models.}
In this work, our focus lies on the algorithmic complexity of fundamental
planning problems and we consider \emph{explicit state-space} graphs, MDPs, and games, where
the complexities are polynomial. The explicit model and
algorithms for it are widely considered:~\cite{LaValle}[Chapter 2.1 Discrete Feasible
Planning],~\cite{KGFP09}[Section IV. DISCRETE
SYNTHESIS] and ~\cite{henzinger2014efficient}[Section 2.1. Definitions. Alternating
game graphs.] 
In other representations such as the factored model, the complexities are higher (NP-complete), and then heuristics are the
focus (e.g.,~\cite{hansen98andor}) rather than the algorithmic complexity. 
Notable exceptions are the work on parameterized complexity of planning problems (see, e.g., \cite{KroneggerPP13})
and Conditional Lower Bounds showing that certain planning problems 
do not admit subexponential time algorithms~\cite{Aghighi2016,backstrom2017time}.

\section{Preliminaries}
\label{sec:preliminaries}
\emph{Markov Decision Processes (MDPs).}
A {\em Markov decision process (MDP)} $P = ((V,E), \ls V_1, V_R \rs, \delta)$ 
consists of a finite set of vertices $V$ partitioned into 
the player-1 vertices $V_1$ and the random vertices $V_R$, 
a finite set of edges 
$E \subseteq (V \times V)$, 
and a probabilistic transition function $\delta$.  
The probabilistic transition function maps every random vertex in $V_R$ to an 
element of $\D(V)$, where $\D(V)$ is the set of probability distributions over 
the set of vertices $V$. A random vertex $v$ has an edge to a
vertex $w \in V$, i.e. $(v,w) \in E$ iff $\delta(v)[w] > 0$. 

\smallskip\noindent\emph{Game Graphs.}
A game graph $\GG = ((V,E), \ls V_1, V_2 \rs)$ consists of a finite set of
vertices $V$, a finite set of edges $E$ and a partition of the vertices $V$ into
player-1 vertices $V_1$ and the adversarial player-2 vertices $V_2$.

\smallskip\noindent{\em Graphs.}
Graphs are a special case of MDPs with $V_R = \emptyset$ as well as special case of 
game graphs with $V_2=\emptyset$. Let $\Out(v)$ describe the set of successor
vertices of $v$. The set $\In(v)$ describes the set of predecessors of the vertex $v$.
More formally $\Out(v) = \{ w \in V \mid (v,w) \in E \}$ and $\In(v) = \{w \in V \mid (w,v) \in E \}$.

\begin{remark}
Note that a standard way to define MDPs is to consider finite vertices with actions, 
and the probabilistic transition function is defined for every vertex and action. 
In our model, the choice of actions is represented as the choice of edges at
player-1 vertices and the probabilistic transition function is represented 
by the random vertices.
This allows us to treat MDPs and game graphs in a uniform way, 
and graphs can be described easily as a special case of MDPs.
\end{remark}

\smallskip\noindent\emph{Plays.}
A {\em play} is an infinite sequence $\omega = \ls v_0, v_1, v_2, \dots \rs$ of
vertices such that each $(v_{i-1},v_i) \in E$ for all $i \geq 1$. 
The set of all plays is denoted with $\Omega$.
A play is initialized by placing a token on an initial vertex. 
If the token is on a vertex owned by a player (such as player~1 in MDPs, or player~1 or player~2 in 
game graphs), then the respective player moves the token along one of the
outgoing edges, whereas if the token is at a random vertex 
$v \in V_R$, then the next vertex is chosen according to the probability 
distribution $\delta(v)$.
Thus an infinite sequence of vertices (or an infinite walk) is formed which is a play.

\smallskip\noindent\emph{Policies.}
Policies are recipes for players to extend finite prefixes of plays.
Formally, a player-i \emph{policy} is a function $\sigma_i: V^* \cdot V_i \mapsto V$ 
which maps every finite prefix $\omega \in V^* \cdot V_i$ of a play that ends in a 
player-i vertex $v$ to a successor vertex $\sigma_i(\omega) \in V$ , i.e., 
$(v, \sigma_i(\omega)) \in E$.
A player-1 policy is \emph{memoryless or stationary} if $\sigma_i(\omega) = \sigma_i(\omega')$ 
for all $\omega, \omega' \in V^* \cdot V_1$ that end in the same vertex $v \in V_1$, i.e., 
the policy does not depend on the entire prefix, but only on the last vertex.

\smallskip\noindent{\em Outcome of policies.} Outcome of policies are as follows:
\begin{compactitem}
\item In graphs, given a starting vertex, a policy for player~1 induces a unique
play in the graph.

\item In game graphs, given a starting vertex $v$, and policies $\sigma_1, \sigma_2$ for 
player~1 and player~2 respectively, the outcome is a unique play $\omega(v,\sigma_1,\sigma_2) 
= \ls v_0, v_1, v_2, \dots \rs$, where $v_0 = v$ and for all $i \geq 0$ if $v_i \in V_1$ then
$\sigma_1(\ls v_0, \dots, v_i\rs) = v_{i+1}$ and if $v_i \in V_2$, then 
$\sigma_2(\ls v_0, \dots, v_i\rs) = v_{i+1}$.

\item In MDPs, given a starting vertex $v$ and a policy $\sigma_1$ for player~1, 
there is a unique probability measure over $\Omega$ which is denoted as $\Pr^\sigma_v(\cdot)$.
\end{compactitem}

\smallskip\noindent\emph{Objectives and winning.}
In general, an \emph{objective} $\phi$ is a measurable subset of $\Omega$.
A play $\omega \in \Omega$ \emph{achieves} the objective if $\omega \in \phi$.
We consider the following notion of winning:
\begin{compactitem}
\item {\em Almost-sure winning.} In MDPs, a player-1 policy $\sigma$ is 
almost-sure (a.s.) winning from a starting vertex $v \in V$ for an objective $\phi$ iff
$\Pr_v^\sigma(\phi) = 1$. 

\item {\em Winning.} In game graphs a policy $\sigma_1$ is \emph{winning for player~1} 
from a starting vertex $v$ iff the resulting play achieves the objective
irrespective of the policy of player~2, i.e., for all $\sigma_2$ we have 
$\omega(v,\sigma_1,\sigma_2) \in \phi$.
\end{compactitem}
Note that in the special case of graphs both of the above winning notions requires that there exists
a play from $v$ that achieves the objective.

\begin{remark}
In MDPs we consider a.s. winning 
for which the precise transition probabilities of the transition
function $\delta$ does not matter, but only the support of the transition
function is relevant. The a.s. winning notion we use corresponds to the strong
cyclic planning problem. Intuitively, if we visit a random vertex in an MDPs
infinitely often then all its successors are visited infinitely often. This
represents the local fairness condition~\cite{Clarke:2000:MC:332656}. Therefore, when we consider the MDP
question only the underlying graph structure along with the partition is
relevant, and the transition function $\delta$ can be treated as a uniform
distribution over the support. 
\end{remark}

We have defined the notion of objectives in general above, and below we 
consider specific objectives that are natural in planning problems. 
They are all variants of one of the most fundamental objectives in computer science, namely, 
reachability objectives.

\smallskip\noindent\emph{Basic Target Reachability.} For a set $T \subseteq V$ of 
target set vertices, the basic target reachability objective is the set of infinite paths that 
contain a vertex of $T$, i.e.,
 $\Reach(T) =\{\ls v_0,v_1, v_2, \dots \rs \in \Omega \mid \exists j \geq 0 : v_j \in T \}$.

\smallskip\noindent\emph{Coverage Objective.}
For $k$ different target sets, namely $T_1, T_2, \dots,T_k$, the coverage 
objective asks whether for each $1 \leq i \leq k$ the basic target reachability objective
$\Reach(T_i)$ can be achieved.
More precisely, given a starting vertex $v$, one asks whether for every $1 \leq i \leq k$ there is a policy $\sigma_1^i$ to ensure winning 
(resp., a.s. winning) for the objective $\Reach(T_i)$ from $v$ for game graphs (resp., MDPs).

\smallskip\noindent
\emph{Sequential Target Reachability.} For a tuple of vertex sets $\T =
(T_1,T_2,\dots,T_k)$ the sequential target reachability objective is the 
set of infinite paths that contain a vertex of $T_1$ followed by a vertex of
$T_2$ and so on up to a vertex of $T_k$, i.e.,
$\Seq(\T) = \{ \ls v_0, v_1, v_2, \dots \rs \in \Omega  \mid \exists j_1, j_2,
\dots j_k : v_{j_1} \in T_1, v_{j_2} \in T_2, \dots, v_{j_k} \in T_k \text{ and }
j_1 \leq j_2 \leq \dots \leq j_k \}$.

\smallskip\noindent\emph{Difference between MDPs and Game Graphs.}
Let the graph $G = (V,E)$ be defined
as follows: Let $V = \{v_1,v_2,v_3\}$ and $E = \{(v_1,v_2), (v_2,v_1),
	(v_2,v_3) \}$. Let $T = \{v_3\}$ be a target set. We will now consider
$\Reach(T)$ for the MDP $P = (G, \langle V_1, V_R \rangle, \delta)$ and the
game graph
$\Gamma = (G, \langle V_1, V_2 \rangle)$. Let $V_1 = \{v_1,v_3\}$ and $V_2 =
V_R =\{v_2\}$. 
The example is illustrated in Figure~\ref{fig:diffgamesmdp}. 
The adversary always chooses to go to $v_1$ and
the target is never reached from $v_1$.
On the other hand, if $v_2$ is probabilistic whenever the token is at $v_2$
it is moved to $v_3$ with non-zero probability.
That is, almost-surely the transition from $v_2$ to $v_3$ is taken
eventually, i.e. $v_3$ is reached almost-surely.
Thus, reachability in MDPs does not imply reachability in game graphs.

\begin{figure}[t]
	\centering
	\includegraphics{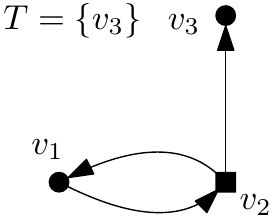}
	\caption{Example illustrating the difference between MDPs and Game
		Graphs for the reachability objective $\Reach(T)$.}
	\label{fig:diffgamesmdp}
\end{figure}

\smallskip\noindent\emph{Relevant parameters.}
We will consider the following parameters: $n$ denotes the number of vertices,
$m$ denotes the number of edges and $k$ will either denote the number of
target sets in the coverage problem or the size 
of the tuple of target sets in the sequential target reachability problem.

\smallskip\noindent{\em Algorithmic study.} In this work we study the above basic 
planning objectives for graphs, game graphs (i.e., winning in game graphs), and 
MDPs (a.s. winning in MDPs). Our goal is to clarify the algorithmic complexity of 
the above questions with improved algorithms and conditional lower bounds.
We define the conjectured lower bounds for conditional lower bounds below.

\subsection{Conjectured Lower Bounds}\label{ss:lowerbounds}
Results from classical complexity are based on standard complexity-theoretical
assumptions, e.g., \textbf{P $\neq$ NP}. Similarly, we derive
polynomial lower bounds which are based on widely believed, conjectured lower
bounds on well studied algorithmic problems. In this work the lower bounds we
derive depend on the popular conjectures below: 

First of all, we consider conjectures on Boolean Matrix
Multiplication~\cite{williams2018subcubic}[Theorem 6.1] and
triangle detection in graphs~\cite{abboud2014popular}[Conjecture 2], which are 
the basis for lower bounds on dense
graphs. A triangle in a graph is a triple $x,y,z$ of vertices such that
$(x,y),(y,z),(z,x) \in E$. We will for the rest of this work assume 
that vertices contain at least one outgoing edge and no self-loops
in instances of Triangle. 
This can be easily established by linear time preprocessing.
See Remark~\ref{remark:combinatorial} for an explanation of the term
``combinatorial algorithm''.

\begin{conjecture}[Combinatorial Boolean Matrix Multiplication Conjecture
(BMM)]\label{conjecture:bmm}
There is no $O(n^{3-\epsilon})$ time combinatorial algorithm for computing the
boolean product of two $n \times n$ matrices for any $\epsilon > 0$.
\end{conjecture}

\begin{conjecture}[Strong Triangle Conjecture (STC)]\label{conjecture:stc}
There is no $O(\min \{ n^{\omega - \epsilon}, m^{2\omega/(\omega+1)-\epsilon} \})$
expected time algorithm and no $O(n^{3-\epsilon})$ time combinatorial algorithm
that can detect whether a graph contains a triangle for any $\epsilon > 0$,
where $\omega < 2.373$ is the matrix multiplication exponent.
\end{conjecture}
\citeauthor{williams2018subcubic}~\shortcite[Theorem 6.1]{williams2010subcubic} showed that BMM is equivalent to the
combinatorial part of STC. Moreover, if we do not restrict ourselves to
combinatorial algorithms, STC, still gives a super-linear lower bound.

\begin{remark}[Combinatorial Algorithms]\label{remark:combinatorial}
``Combinatorial'' in Conjecture~\ref{conjecture:stc} means that it excludes ``algebraic methods'' (such as fast matrix
multiplication~\cite{williams2012multiplying,le2014powers}), which are
impractical due to high associated constants. Therefore the term ``combinatorial
algorithm'' comprises only discrete algorithms. Non-combinatorial algorithms usually have
the matrix multiplication exponent $\omega$ in the running time.
Notice that all algorithms for deciding almost-sure winning conditions in MDPs and winning conditions in games are
discrete graph-theoretic algorithms and hence are combinatorial, and thus lower bounds for 
combinatorial algorithms are of particular interest in our setting.
For further discussion consider~\cite{BallardDHS12,henzinger2015hardness}.
\end{remark}

\smallskip\noindent
Secondly, we consider the Strong Exponential Time Hypothesis
(SETH) used also in~\cite{abboud2014popular}[Conjecture 1] introduced by~\cite{impagliazzo1999complexity,impagliazzo1998problems} for the satisfiability problem of propositional logic 
and the Orthogonal Vector Conjecture. 

\smallskip\noindent
\emph{The Orthogonal Vectors Problem (OV).} Given sets $S_1,S_2$ of
$d$-bit vectors with $|S_1| =|S_2| = N$ and $d = \omega(\log N)$, are there $u \in
S_1$ and $v \in S_2$ such that $\sum_{i=1}^d u_i \cdot v_i = 0$? 

\smallskip\noindent
\begin{conjecture}[Strong Exponential Time Hypothesis (SETH)] For each
$\epsilon>0$ there is a $k$ such that $k$-CNF-SAT on $n$ variables and $m$
clauses cannot be solved in $O(2^{(1-\epsilon)n} \poly(m))$ time.
\end{conjecture}

\begin{conjecture}[Orthogonal Vectors Conjecture (OVC))]\label{conj:ovc}
There is no $O(N^{2-\epsilon})$ time algorithm for the Orthogonal Vectors
Problem for any $\epsilon > 0$.
\end{conjecture}

\citeauthor{williams2005satisfaction}~\shortcite{williams2005satisfaction}[Theorem 5] 
SETH implies OVC, which is an implications of a result in~\cite{williams2005satisfaction} and 
an explicit reduction is given in the survey article by \citeauthor{VW2018survey}  \shortcite[Theorem 3.1]{VW2018survey}.
Whenever a problem is provably hard assuming OVC it is thus also hard when assuming
SETH. For example, in~\cite{bringmann2015quadratic}[Preliminaries, A.
Hardness Assumptions, OVH] the OVC is assumed to prove conditional lower bounds for the longest
common subsequence problem. To the best of the author's knowledge, there is no connection between the former two and the latter two
conjectures.

\begin{remark}
The conjectures that no polynomial improvements over the best-known running
times are possible do not exclude improvements by sub-polynomial factors such as
polylogarithmic factors or factors of, e.g., $2^{\sqrt{\log n}}$.
\end{remark}

\section{Basic Previous Results}\label{sec:bresults}
In this section, we recall the basic algorithmic results about MDPs and game graphs
known in the literature that we later use in our algorithms. 

\smallskip\noindent\emph{Basic result~1: Maximal End-Component Decomposition.}
Given an MDP $P$, an {\em end-component} is a set of vertices $X \subseteq V$ s.t. 
(1)~the subgraph induced by $X$ is strongly connected (i.e., $(X,E \cap X \times X)$ 
is strongly connected) and 
(2)~all random vertices have their outgoing edges in $X$, i.e., $X$ is closed for 
random vertices, formally described as: for all $v \in X \cap V_R$ and all 
$(v,u) \in E$ we have $u \in X$. 
A \emph{maximal end-component} (MEC) is an end-component which is maximal 
under set inclusion. 
The importance of MECs is as follows:
(i)~first it generalizes strongly connected components (SCCs) in graphs (with $V_R = \emptyset$)
and closed recurrent sets of Markov chains (with $V_1=\emptyset$); and 
(ii)~in a MEC $X$ from all vertices $u \in X$ every vertex $v \in X$ can be reached
almost-surely.
The MEC-decomposition of an MDP is the partition of the vertex set into MECs and 
the set of vertices which do not belong to any MEC. 
While MEC-decomposition generalizes SCC decomposition of graphs, 
and SCC decomposition can be computed in linear time~\cite[Theorem 13]{tarjan1972},
there is no linear-time algorithm for MEC-decomposition computation.
The current best-known algorithmic bound for MEC-decomposition is 
$O(\min(n^2, m^{1.5})= O(m \cdot n^{2/3}))$~\cite[Theorem 3.6, Theorem 3.10]{henzinger2014efficient}.

\smallskip\noindent\emph{Basic result~2: Reachability in MDPs.}
Given an MDP $P$ and a target set $T$, the set of starting vertices
from which $T$ can be reached almost-surely can be computed 
in $O(m)$ time given the MEC-decomposition of $P$~\cite[Theorem 4.1]{chatterjee2016separation}.
Moreover, for the basic target reachability problem the current best-known 
algorithmic bounds are the same as the MEC-decomposition problem, i.e.,
$O(\min(n^2, m^{1.5}))= O(m \cdot n^{2/3})$~\cite[Theorem 3.6, Theorem 3.10]{henzinger2014efficient},
and any improvement for the MEC-decomposition algorithm also carries
over to the basic target reachability problem.

\smallskip\noindent\emph{Basic result~3: Reachability in game graphs.}
Given a game graph $\Gamma$ and a target set $T$, 
the set of starting vertices from which player~1 can ensure to reach $T$
against all polices of player~2, is called {\em player-1 attractor}
to $T$ and can be computed in $O(m)$ time~\cite{beeri1980membership,immerman1981number}.

The above basic results from the literature explain the result of the 
first row of Table~\ref{tab:complexity}.

\section{Coverage Problem}\label{sec:utarget}
In this section, we consider the coverage problem.
First, we present the algorithms, which are simple, and then focus on 
the conditional lower bounds for MDPs and game graphs, which establish
that the existing algorithms cannot be (polynomially) improved under the STC and OV 
conjectures.

\subsection{Algorithms}
We present a linear-time algorithm for graphs, and quadratic time 
algorithm for MDPs and game graphs. 
The results below present the upper bounds of the second row of 
Table~\ref{tab:complexity}.

\smallskip\noindent\emph{Planning in Graphs.}
For the coverage problem in graphs we are given a graph $G = (V,E)$, 
a vertex $s \in V$ and target sets $T_1, T_2,\dots, T_k$.
The algorithmic problem is to find out if starting from an initial 
vertex $v$ the basic target reachability, i.e., $\Reach(T_i)$, can be achieved for 
all $1 \leq i \leq k$. 
The algorithmic solution is as follows:
Compute the BFS tree starting from $s$ and check if all the targets 
are contained in the resulting BFS tree.

\smallskip\noindent\emph{Planning in MDPs and Games.}
For both MDPs and game graphs with $k$ target sets, the basic algorithm 
performs $k$ basic reachability computations, i.e., 
for each target set $T_i$, $1 \leq i \leq k$, the basic target reachability 
for target set $T_i$ is computed.
(1)~For game graphs, using the $O(m)$-time attractor computation (see Basic result~3), 
we have an $O(k\cdot m)$-time algorithm.
(2)~For MDPs, the MEC-decomposition followed by $k$ many $O(m)$-time almost-sure 
reachability computation (see Basic result~2), 
gives an  $O(k \cdot m + \mectime)$ time algorithm.

\subsection{Conditional Lower Bounds}
We present conditional lower bounds for the coverage problem in MDPs and game
graphs (i.e., the CLBs of the second row of Table~\ref{tab:complexity}).  
For MDPs and game graphs the conditional lower bounds complement the 
quadratic algorithms from the previous subsection. 
The conditional lower bounds are due to reductions from OV and Triangle.

\smallskip\noindent\emph{Sparse MDPs.}
For sparse MDPs we present a conditional lower bound based on OVC.
To do that we reduce the OV problem to the coverage problem in MDPs.

\begin{reduction}\label{red:query_mdp_ov}
Given two sets $S_1, S_2$ of $d$-dimensional vectors, we build the MDP
$P$ as follows.
 \begin{compactitem}
 	\item 
	The vertices $V$ of the MDP are given by a start vertex $s$, 
	sets of vertices $S_1$ and $S_2$ representing the sets of vectors 
	and vertices $\C = \{c_i \mid 1 \leq i \leq d\}$ representing the coordinates of the vectors in the OVC instance. 
 
 	\item 
	The edges $E$ of $P$ are defined as follows: The start vertex $s$ has an edge to every vertex of
	$S_1$. Furthermore for each $x_i \in S_1$ there is an edge to $c_j \in C$ iff 
	$x_i[j] = 1$ and for each $y_i \in S_2$ there is an edge from $c_j \in S_2$ to
	$y_i$ iff $y_i[j] = 1$.
 
 	\item 
	The set of vertices is partitioned into player-1 vertices $V_1  = S_1
	\cup \C \cup S_2$ and random vertices $V_R  = \{s\}$.
 \end{compactitem}
\end{reduction}

\begin{figure}[t]
  \centering
  \includegraphics{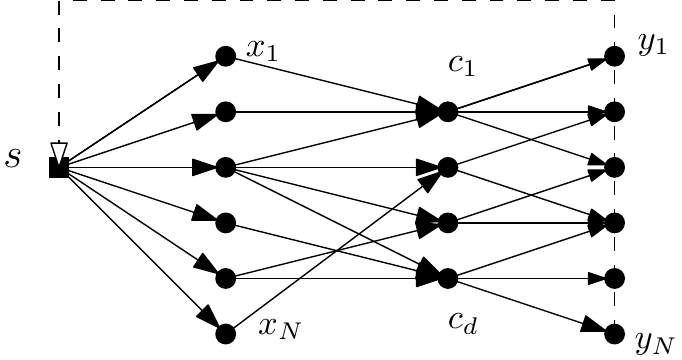}
  \caption{Reduction from OV}
  \label{fig:seq_games_ov}
\end{figure}

The reduction is illustrated in Figure~\ref{fig:seq_games_ov} (the dashed edges will be used later for the sequential target lower bounds).

\begin{lemma}
Let $P = (V, E, \langle V_1, V_R \rangle, \delta)$ 
be the MDP given by Reduction~\ref{red:query_mdp_ov} with target sets $T_i = \{y_i\}$ for $i = 1\dots N$.
There exist orthogonal vectors $x \in S_1$, $y \in S_2$ iff there is no a.s. winning policy from $s$
for the coverage objective.
\end{lemma}

\begin{proof}
  Notice that when starting from $s$ the token is randomly moved to one of the vertices $x_i$ and thus 
  player~1 can reach each $y_j$ almost surely from $s$ iff it can reach each $y_j$ from each $x_i$.
  The MDP $P$ is constructed in such a way that there is no
  path between vertex $x_i$ and $y_j$ iff the corresponding vectors are
  orthogonal in the OV instance: If $x_i$ is orthogonal to $y_j$, the outgoing
  edges lead to no vertex which has an incoming edge to $y_j$ as either $x_i[k]
  = 0$ or $y_j[k] = 0$. One the other hand, if there is no path from $x_i$ to
  $y_j$ we again have by the construction of the underlying graph that for all
  $1 \leq k \leq d: x_i[k] = 0$ or $y_j[k] =0$. This is the definition of
  orthogonality for $x_i$ and $y_j$.
  Thus, player~1 can reach all the target sets a.s. 
  from $s$ iff there are no orthogonal vectors in $S_1$ and $S_2$.
\end{proof}

  The MDP $P$ has only $O(N)$ many vertices and Reduction~\ref{red:query_mdp_ov} can
  be performed in 
  $O(N \log N)$ time (recall that $d = \omega(\log N)$). The number of edges $m$ is $O(N \log N)$ 
  and the number of target sets $k \in
  \theta(N)$. Thus the theorem below follows immediately.

\begin{theorem}\label{thm:cover:mdpsparse}
  There is no $O(m^{2-\epsilon})$ or $O((k\cdot m)^{1-\epsilon})$ (for any $\epsilon > 0$) algorithm to check 
  if a vertex $v$ has an a.s. winning policy for the coverage problem in MDPs under Conjecture~\ref{conj:ovc} (i.e., unless OVC and SETH fail). 
\end{theorem}

\smallskip\noindent\emph{Dense MDPs.}
For dense MDPs we present a conditional lower bound based on boolean matrix
multiplication (BMM).
Therefore we reduce the Triangle problem to the coverage problem in MDPs.

\begin{reduction}\label{red:query_mdp_triangle}
  Given an instance of triangle detection, i.e., a graph $G = (V,E)$, we build the
  following MDP $P = (V',E',\ls V'_1, V'_R \rs, \delta)$.
   \begin{compactitem}
 	  \item 
	  The vertices $V'$ are given as four copies $V_1,V_2,V_3,V_4$ of $V$
	  and a start vertex $s$. 
 	  \item 
	  The edges $E'$ of $P$ are defined as follows: There
	  is an edge from $s$ to every $v_{1i} \in V_1$ for $i = 1 \dots n$. In
	  addition for $1 \leq j \leq 3$ there is an edge from $v_{ji}$ to $v_{(j+1)k}$ iff
	  $(v_i,v_k) \in E$.
 
 	  \item 
	  The set of vertices $V'$ is partitioned into player-1 vertices $V'_1 =
	  \emptyset$ and random vertices $V'_R = \{s\} \cup V_1 \cup V_2 \cup V_3 \cup
	  V_4$.
   \end{compactitem}
\end{reduction}

\begin{figure}
  \includegraphics{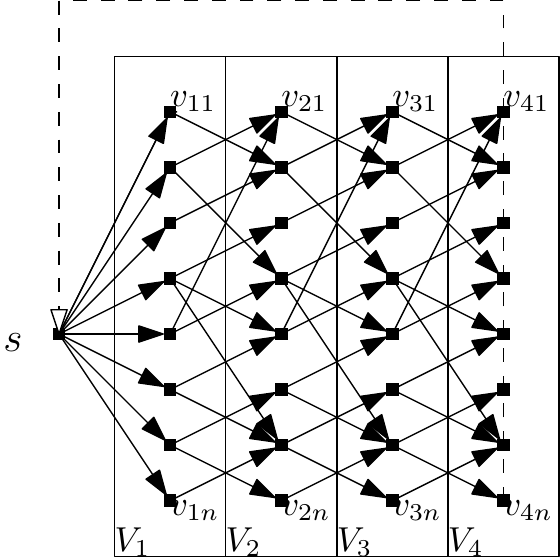}
  \caption{Reduction from Triangle}
  \label{fig:seq_triangle_games}
\end{figure}

The reduction is illustrated in Figure~\ref{fig:seq_triangle_games} (the dashed
edges will be used later for the sequential target lower bounds).

\begin{lemma}
Let $P$ be the MDP given by Reduction~\ref{red:query_mdp_triangle}
with $n$ target sets $T_1,\dots, T_n$. The target set $T_i = V_1 \setminus \{v_{1i}\} \cup
V_4 \setminus \{v_{4i}\}$ for $i = 1 \dots n$.  
A graph $G$ has a triangle iff player-1 has an a.s. winning policy from $v$ for
the coverage objective.
\end{lemma}

\begin{proof}
Notice that there is a triangle in the graph $G$ iff there is a path from some vertex
$v_{1i}$ in the first copy of $G$ to the same vertex in the fourth copy of $G$,
$v_{4i}$. Also, a path starting in $s$ satisfies the coverage objective, i.e., reaches all target
sets a.s., unless it visits a vertex $v_{1i}$ and also $v_{4i}$.
As each of these paths has non-zero probability player~1 wins almost-surely from $v$ iff
there is no such path iff there is no triangle in the original graph.
\end{proof}

Moreover, the size and the construction time of the MDP $P$
are linear in the size of the
original graph $G$ and we have $k = \theta(n)$ target sets.
Thus the theorem below follows immediately.

\begin{theorem}\label{thm:cover:mdpsdense}
There is no combinatorial $O(n^{3-\epsilon})$ or $O((k\cdot n^2)^{1-\epsilon})$
algorithm (for any $\epsilon > 0$) to check if a vertex has an a.s. winning
policy for the coverage objective in MDPs 
under Conjecture~\ref{conjecture:stc} (i.e., unless STC and BMM fail).  
The bounds hold for dense MDPs with $m = \theta(n^2)$.
\end{theorem}

Next, we describe how the results for MDPs can be extended to game graphs.

\smallskip\noindent\emph{Sparse Game Graphs.}
The random starting vertex in the reduction is changed to a player-2 vertex. The
rest of the reduction stays the same. 
The proof then proceeds as before with the adversary player~2 now overtaking the role of the random choices.

\begin{reduction}\label{red:query_games_ov}
Given two sets $S_1, S_2$ of $d$-dimensional vectors, we build the following
game graph $\Gamma = (V,E,\ls V_1, V_2 \rs)$.
	\begin{compactitem}
	\item The vertices $V$ of the game graph are given by a start vertex $s$, 
		sets of vertices $S_1$ and $S_2$ representing the sets of vectors 
		and vertices $\C = \{c_i \mid 1 \leq i \leq d\}$ representing the
		coordinates. 

	\item The edges $E$ of $\Gamma$ are defined as follows: the start
		vertex $s$ has an edge to every vertex of
		$S_1$. Furthermore for each $x_i \in S_1$ there is an edge
		to $c_j \in C$ iff 
		$x_i[j] = 1$ and for each $y_i \in S_2$ there is an edge
		from $c_j \in S_2$ to
		$y$ iff $y_i[j] = 1$.

	\item The set of vertices is partitioned into
		player-1 vertices $V_1  = S_1
		\cup \C \cup S_2$ and player-2 vertices $V_2
		= \{s\}$.
	\end{compactitem}
\end{reduction}

The reduction is illustrated in Figure~\ref{fig:seq_games_ov} (the dashed edges will be used later for the sequential target lower bounds).

\begin{lemma}
Let $\Gamma$ be the game graph given by Reduction~\ref{red:query_games_ov} with target sets 
target sets $T_i = \{y_i\}$ for $i = 1\dots N$.
There exist orthogonal vectors $x \in S_1$, $y \in S_2$ iff there is no
winning policy from start vertex $s$ for the coverage objective.
\end{lemma}

\begin{proof}
  Notice that when starting from $s$ the token is moved to one of the vertices $x_i$ and thus 
  player~1 can reach each $y_j$ from $s$ iff it can reach each $y_j$ from each $x_i$.
  If there is one $y_j$ which cannot be reached from an $x_i$, player~2 will
  choose $x_i$ as successor and win.
  The game graph $\Gamma$ is constructed in such a way that there is no
  path between vertex $x_i$ and $y_j$ iff the corresponding vectors are
  orthogonal in the OV instance: If $x_i$ is orthogonal to $y_j$, the outgoing
  edges lead to no vertex which has an incoming edge to $y_j$ as either $x_i[k]
  = 0$ or $y_j[k] = 0$. One the other hand, if there is no path from $x_i$ to
  $y_j$ we again have by the construction of the underlying graph that for all
  $1 \leq k \leq d: x_i[k] = 0$ or $y_j[k] =0$. This is the definition of
  orthogonality for $x_i$ and $y_j$.
  Thus, player~1 can reach all the target sets from starting vertex $s$ 
  iff there are no orthogonal vectors in $S_1$ and $S_2$.
\end{proof}

  The game graph $\Gamma$ has only $O(N)$ many vertices and Reduction~\ref{red:query_mdp_ov} can
  be performed in 
  $O(N \log N)$ time (recall that $d = \omega(\log N)$). The number of edges $m$ is $O(N \log N)$ 
  and the number of target sets $k \in
  \theta(N)$. Thus the theorem below follows immediately.

\begin{theorem}\label{thm:cover:gamessparse}
There is no $O(m^{2-\epsilon})$ or $O((k\cdot m)^{1-\epsilon})$ algorithm (for
any $\epsilon > 0$) to check if a vertex has a winning policy for 
the coverage objective with $k$ reachability objectives in
game graphs under Conjecture~\ref{conj:ovc} (i.e., unless OVC and SETH fail). 
\end{theorem}

\smallskip\noindent\emph{Dense Game Graphs.}
The random vertices in the reduction are now player-2 vertices. Notice that the
resulting game graph $\Gamma$ has only player-2 vertices. 
Now if there is a path starting from $s$ that is not in the defined
coverage objective then player~2
would simply choose that one and thus player~1 still wins iff there is no such path, i.e.,
there is no triangle in the original graph.

\begin{reduction}\label{red:query_games_triangle}
Given an instance of triangle detection, i.e., a graph $G = (V,E)$, we build the
following game graph $\Gamma = (V',E',\ls V'_1, V'_2 \rs)$.
\begin{compactitem}
\item The vertices $V'$ are given as four copies $V_1,V_2,V_3,V_4$ of $V$
and a start vertex $s$. 
\item The edges $E'$ are defined as follows: There
is an edge from $s$ to every $v_{1i} \in V_1$ for $i=1\dots n$. In
addition for $1 \leq j \leq 3$ there is an edge from $v_{ji}$ to
$v_{(j+1)k}$ iff
$(v_i,v_k) \in E$. 
\item The set of vertices $V'$ is partitioned into
player-1 vertices $V'_1 =
\emptyset$ and player-2 vertices $V'_2 = \{s\} \cup
V_1 \cup V_2 \cup V_3 \cup
V_4$.
\end{compactitem}
\end{reduction}

\begin{lemma}
Let $\Gamma$ be the game graph given by Reduction~\ref{red:query_games_triangle}
with $n$ target sets $T_1,\dots, T_n$. The target set $T_i = V_1 \setminus \{v_{1i}\} \cup
V_4 \setminus \{v_{4i}\}$ for $i = 1 \dots n$.  
A graph $G$ has a triangle iff player~1 has a winning policy from $s$ for
the coverage objective.
\end{lemma}

\begin{proof}
Notice that there is a triangle in the graph $G$ iff there is a path from some vertex
$v_{1i}$ in the first copy of $G$ to the same vertex in the fourth copy of $G$,
$v_{4i}$. Also, a path starting in $s$ satisfies the coverage objective, i.e., reaches all target
sets unless it visits a vertex $v_{1i}$ and also $v_{4i}$.
Player~1 wins from $v$ iff there is no such path as player 2 choose it. 
Such a path exists as proved above iff there is no triangle in the original graph.
\end{proof}

Moreover, the size and the construction time of game graph $\Gamma$
are linear in the size of the
original graph $G$ and we have $k = \theta(n)$ target sets.
Thus the theorem below follows immediately.

\begin{theorem}\label{thm:cover:gamesdense}
There is no combinatorial $O(n^{3-\epsilon})$ or $O((k\cdot n^2)^{1-\epsilon})$
algorithm (for any $\epsilon > 0$) to check whether a vertex $v$ has a winning
policy for the coverage objective 
in game graphs under Conjecture~\ref{conjecture:stc} (i.e., unless STC and BMM fail). 
The bounds hold for dense game graphs with $m = \theta(n^2)$.
\end{theorem}

\section{Sequential Target Problem}
We consider the sequential target problem in graphs, MDPs and game graphs. 
In contrast to the quadratic CLB for the coverage problem, 
quite surprisingly we present a subquadratic algorithm for MDPs,
which as a special case gives a linear-time algorithm for graphs.
For games, we present a quadratic algorithm and a quadratic CLB.

\subsection{Algorithms}
The results below present the upper bounds of the third row of
Table~\ref{tab:complexity}.

\begin{algorithm}[ht]\label{alg:seq_mdps}
\KwIn{MEC-free MDP $P = (V,E, \ls V_1, V_R \rs, \delta)$ and a tuple of target
sets $\T = (T_1,\dots,T_k)$.}
\KwOut{All vertices with a policy for $\Seq(\T)$.} 
$L_v \gets \{ i \mid v \in T_i : i = 1 \dots k \} \ \forall v \in V$\;
$A[i] \gets 0$  for $1 \leq i \leq k$\;
$\mcount_v \gets$ number of outgoing edges of $v$\;
$best_v \gets \nill$, $\ell_v \gets \nill$ for $v \in V$\;
$best_v \gets k+1$ for $v \in V$ with no outgoing edges\;
$S \gets V$\; 
Queue $Q \gets \{v \in V \mid \text{ v has no outgoing edges} \}$\;
\BlankLine

\While{$S \not= \emptyset$\label{alg:seq_mdps:whileloop}}
  {
	\If{$Q \neq \emptyset$} {
		v = Q.pop()\;
		\texttt{ProcessVertex}(v)\;
	} 
	\Else{
		$v \gets \argmax_{v \in V_R \cap S}\ best_v$\label{alg:seq_mdps:argmax}\;
		\texttt{ProcessVertex}(v)\;
	}
	
  }
\Return $\{v \in V \mid \ell_v=1\}$\;
 \BlankLine
 \BlankLine
\Function{\texttt{\em ProcessVertex(Vertex v)}}{   
     \lFor{$i \in L_v$}{ $A[i] \gets 1$ }
     $\ell_v \gets best_v $\label{alg:seq_mdps:l_v}\;
     \While{$A[\ell_v -1 ] = 1 \land \ell_v > 1$\label{alg:seq_mdps:computeEllv}}
	  {
	    $\ell_v \gets \ell_v-1$\;
	  }
     \lFor{$i \in L_v$}{$A[i] \gets 0$}
    
    $S \gets S \setminus \{v\}$\label{alg:seq_mdps:updateS}\; 
    \For{$w \in \{w : (w,v) \in E\}$}
      {
	\uIf{$w \in V_1$\label{alg:seq_mdps:updatebest}}
	  {
	    $\best_w \gets \min(\best_w,\ell_v)$\label{alg:seq_mdps:updatebest1}
	  }
	\Else
	  {
	    $\best_w \gets \max(\best_w,\ell_v)$\label{alg:seq_mdps:updatebest2}
	  }
	$\mcount_w \gets \mcount_w -1$\label{alg:seq_mdps:decrcounter}\;
	\uIf{$\mcount_w =0 \land w \in S$}
	  {
	    $Q.push(w)$\label{alg:seq_mdps:pushr}
	  }
      }
}
\caption{Sequential target Reachability for MEC-free MDPs.}
\end{algorithm}

\smallskip\noindent\emph{Planning in MDPs.}
We first calculate the MEC-decomposition of the MDP. Then each
MECs is collapsed into a single vertex which we set to be a player-1 vertex. 
In the target sets, all the vertices of the MEC are replaced by this new vertex.
This does not change the reachability conditions of the resulting MDP: 
Every vertex in th MEC can be
reached almost surely starting from every
other vertex in the same MEC, regardless of their type (player-1, random).
Thus it suffices to give an algorithm for  a MEC-free MDP $P = (V,E, \langle V_1, V_R
\rangle, \delta)$ with tuple of target sets $(T_1,\dots, T_k)$. 

The vertices in $S$ are the vertices that are not processed yet and $S$ is
initialized with $V$.
Initially, vertices with no outgoing edges are added to a queue $Q$.
Throughout the algorithm, the queue $Q$ contains the vertices which have not been processed so far but
whose successors are already processed.

While the queue $Q$ is not empty, a vertex from the queue is processed. 
When a vertex $v$ is processed the function $\texttt{ProcessVertex}(v)$ is
called. The function calculates the label $\ell_v$ of the vertex $v$
and updates variables $best_w$ and $\mcount_w$ of the other vertices. 
The label $\ell_v$ means vertex $v$ has an almost-sure 
winning policy for the objective $\Seq(\T_{\ell_v})$ where $\T_{\ell_v} =
(T_{\ell_v}, \dots, T_k)$.
Note that this means that vertices with label $1$ have an almost-sure winning policy for the objective
$\Seq(\T)$ where $\T = \{T_1,T_2,\dots T_k)$.
The variables $best_v$ are used to store the maximum (for $v \in V_R$) / minimum (for $v \in V_1$) 
label of the already processed successors of $v$. 

Now when $Q$ is empty, then the algorithm has to process a vertex where not all successors
have been processed yet. In that case, one considers all the random vertices for which at least
one successor has already been processed and chooses the random vertex with maximum $best_v$ to process next. 
Notice that the function $\argmax$ ignores arguments with $\nill$ values.
One can show that, as the graph has no MECs, whenever $Q$ is empty (and $S$ is not) there
exist such a random vertex. 
Moreover, whenever $Q$ is empty, all vertices in the set of unprocessed vertices $S$ 
have a policy that satisfies $\Seq(\T_{m})$ for $m = \max_{v \in V_R \cap S}\ best_v$.
Intuitively, this is due to the fact that all vertices $v \in S$ can reach the set of already
processed vertices and in the worst case the reached vertex $v'$ has $\ell_{v'}=m$ and thus a strategy for $\Seq(\T_{m})$.
For the selected vertex $v$ all its successors $w$ will, thus, finally have a label $\ell_w$ 
of at most $m$, and, as the current value of $\best_v$ is $m$, there is a successor $w$ with $\ell_w=m$. 
Thus, as $v \in V_R$, we have that also the final value of $best_v$ must be $m$. 
Hence, one can already process $v$ without knowing the labels of all the successors.

\begin{prop}[Correctness]
Given an MDP $P$ and a sequential target objective $\Seq(\T)$ with targets $\T =
\{T_1, \dots, T_k\}$, Algorithm~\ref{alg:seq_mdps} decides whether there is a
player-1 policy at a start vertex $s$ for the objective $\Seq(\T)$.
\end{prop}

We next state invariants of the while loop (see Line~\ref{alg:seq_mdps:whileloop}) that we will use later 
to show a loop invariant that will establish the correctness of the algorithm.

\begin{lemma}\label{lem:seq_mdps:invariants1}
  The following statements are invariants of the while loop in
  Line~\ref{alg:seq_mdps:whileloop}.
  \begin{enumerate}
      \item $\mcount_v= |Out(v) \cap S|$;
      \item $v \in Q$ iff $v \in S$ and $\Out(v) \cap S = \emptyset$;
	  \item $best_v = k+1$, for all $v \in V$ with $\Out(v)= \emptyset$.
	  \item $best_v = 
			  \begin{cases}
				  \min_{w \in \Out(v) \setminus S} \ell_w & v \in V_1 \\
				  \max_{w \in \Out(v) \setminus S} \ell_w & v \in V_R
			  \end{cases}$, for all $v \in V$ with $\Out(v) \not= \emptyset$;\\
			  In particular $\ell_w\not=\nill$ for all $w \in V \setminus S$.
	  \item If $S\not= \emptyset$ and $Q=\emptyset$ there is a $v\in S \cap V_R$ such that $best_v \not= \nill$.
  \end{enumerate}
\end{lemma}
The above invariants state that (a) the variables have the intended meaning,
(b) $Q$ contains all the unprocessed vertices whose successors are already processed, and 
(c) that the function $\argmax$ is well-defined whenever called.
These are three important ingredients to show the correctness of Algorithm~\ref{alg:seq_mdps}.

\begin{proof}
\begin{enumerate}
    \item The counters $\mcount_v$ are initialized as  $|Out(v)|$ and
		  $S$ is initialized as $V$. Thus the claim holds when first entering the while loop.
		  
		  Assume the claim holds at the beginning of the iteration where vertex $v$ is processed. 
		  The set $S$ is only changed in Line~\ref{alg:seq_mdps:updateS}
		  where only $v$ is removed from the set while
		  the counters are only changed in Line~\ref{alg:seq_mdps:decrcounter},
		  where all counters of vertices $w$ with $v \in \Out(w)$ are decreased by one
		  (notice that $v \in \Out(w)$ iff $w \in \In(v)$).
		  That is, $\mcount_v= |Out(v) \cap S|$ also after this iteration of the loop
		  and the claim follows.
	\item In the initial phase $S$ is set to $V$ and $Q$ is set to $\{v \in V \mid \Out(v)=\emptyset\}$. 
		  Thus the claim holds when first entering the while loop.
		  
		  Assume the claim holds at the beginning of the iteration where vertex $v$ is processed. 
		  The set $S$ is only changed in Line~\ref{alg:seq_mdps:updateS}
		  where $v$ is removed.
		  
		  First consider a vertex $w \in Q\setminus \{v\}$. As $w$ is not removed from the set $S$ and no vertex is added to $S$
		  the claim is still true for $w$.
		  Now consider a vertex $w \in S \setminus Q$ that might be added during the iteration of the loop.
		  This can only happen in Line~\ref{alg:seq_mdps:pushr} and the if conditions 
		  ensures that $w \in S$ and $\Out(v) \cap S = \emptyset$ (by the previous invariant).
		  Thus the claim also holds for the newly added vertices.
	\item For $v \in V$ with $\Out(v)= \emptyset$ the variables $best_v$ are initialized with $k+1$ and
		  $best_v$ is only changed in Line~\ref{alg:seq_mdps:updatebest1} or Line~\ref{alg:seq_mdps:updatebest2},
		  when a successor of the vertex is processed.
		  As $v$ has no successor, $best_v$ is not changed during the algorithm.  
	\item Note that we define the $\max$ or $\min$ over the empty set to be null.
		  As all $\ell_v$ are initialized as null and $best_v$ with $\Out(v) \not= \emptyset$ are initialized as $\nill$
		  the claim holds when the algorithm enters the loop.
		  
		  Now consider the iteration of vertex $v$ and assume the claim is true at the beginning.
		  The set $S$ is only changed in Line~\ref{alg:seq_mdps:updateS} where $v$ is removed.
		  Let $S_{old}$ be the set at the beginning of the iteration and 
		  $S_{new} = S_{old} \setminus \{v\}$ the updated set.
		  First notice that $best_v \not= \nill$ as it is either chosen by
		  (a) as element of $Q$ or
		  (b) by $\argmax$.
		  In the former case it was either initially set to $k+1$ or 
		  it was added to $Q$ when processing a vertex $w \in \Out(v)$, 
		  which would have set $best_v$.
		  In the latter case $best_v \not= \nill$ by the definition of $\argmax$.
		  Now, as $best_v \not= \nill$ the assignment in Line~\ref{alg:seq_mdps:l_v} ensures that also 
		  $\ell_v \not= \nill$.
		  For a vertex $w \in \In(v) \cap V_1$ the value $best_w$ is updated to $\min(\best_w,\ell_v)$ (Line~\ref{alg:seq_mdps:updatebest1})
		  which by assumption is equal to $\min_{x \in (\Out(w) \setminus S_{old}) \cup v} \ell_x = 
		  \min_{x \in (\Out(w) \setminus S_{new})} \ell_x$, i.e., the equation holds.
		  For a vertex $w \in \In(v) \cap V_R$  the value $best_w$ is 
		  updated to $\max(\best_w,\ell_v)$ (Line~\ref{alg:seq_mdps:updatebest2})
		  which by assumption is equal to 
		  $\max_{x \in (\Out(w) \setminus S_{old}) \cup v} \ell_x = 
		  \max_{x \in (\Out(w) \setminus S_{new})} \ell_x$, i.e., the equation holds.
		  For vertices $w \notin \In(v)$ 
		  both $best_w$ as well as the right hand side of the equation are unchanged.
		  Hence, the claim holds also after the iteration.
	\item The input graph has a vertex $v$ with $\Out(v)= \emptyset$.
	      Towards a contradiction assume that no such
	      vertex exists. Then an SCC $C$ where every vertex in $C$ has only edges to
	      other vertices in $C$ exists. Such SCCs are called \emph{bottom SCCs}. Bottom SCCs are MECs~\cite{henzinger2014efficient} and
	      we assumed that there are no MECs in the MDP, a contradiction.
	      Thus, $Q$ is non-empty after the initialization and
	      the claim holds after the initialization.	
		  Now consider the iteration of vertex $v$ and assume the claim is true at the beginning
		  and $Q=\emptyset$.
		  Notice that $best_w$ is set for vertices as soon as one vertex $w \in \Out(v)$ was processed.
		  Towards a contradiction assume that all vertices in $S \cap V_R$ have $best_w=\nill$, 
		  i.e., no vertex $v \in S \cap V_R$ has a successor in $V \setminus S$.
		  Each $v \in S$ has at least one successor in $S$ as otherwise, $v$ would be in $Q$.
		  That is $S$ is either empty or has a non-trivial bottom SCC that has no random outgoing edges. 
		  Again such an SCC would be a MEC and thus we obtain our desired contradiction.
\end{enumerate}
\end{proof}

From the following invariant, we obtain the correctness of our algorithm.

\begin{lemma}
	The following statements are invariants of the while loop in Line~\ref{alg:seq_mdps:whileloop} or all $v \in V \setminus S$:
	\begin{enumerate}
		\item there exists a player~1 policy $\sigma$ s.t. $\Pr_v^\sigma(\Seq(\T_{\ell_v})) = 1$; and
		\item there is no player~1 policy $\sigma$ s.t. $\Pr_v^\sigma(\Seq(\T_{\ell_v-1})) = 1$. 
	\end{enumerate}
	where $\T_{\ell_v} = \{ T_{\ell_v}, \dots, T_k\}$ or $\ell_v > k$.
\end{lemma}
\begin{proof}
	As $S$ is initialized as set $V$ the two statements hold after the initialization.
	
	Now consider the iteration where vertex $v$ is processed and assume the invariants hold at the beginning of the iteration.
	We first introduce the following notation 
	\[
		be(v)=
		\begin{cases}
					\min_{w \in \Out(v)} \ell_w & v \in V_1 \\
					\max_{w \in \Out(v)} \ell_w & v \in V_R
		\end{cases}
	\]
	We distinguish the case where $Q$ is non-empty and the case where $Q$ is empty.
	\begin{itemize}
		\item \emph{Case $Q  \not= \emptyset$}:
			By Lemma~\ref{lem:seq_mdps:invariants1} we have  $\Out(v) \cap S = \emptyset$
			and thus also  $\ell_w \not= \nill$ for all $w \in \Out(v)$.
			Thus by Lemma~\ref{lem:seq_mdps:invariants1}(4) 
			we have $be(v)= best_v$ and $\ell_v$ can be computed.
			By the while loop in Line~\ref{alg:seq_mdps:computeEllv} 
			we have $L_v$ contains $\ell_v, \dots, best_v-1$ but does not contain $\ell_{v}-1$.

			1) Thus we can easily obtain a policy $\sigma$ 
			with $\Pr_v^\sigma(\Seq(\T_{\ell_v})) = 1$ as follows.
			If $v \in V_1$ pick the vertex $w$ that corresponds to $be(v)$ and
			then player~1 can follow the
			existing policy for vertex $w$. 
			If $v \in V_R$ 
			then which ever vertex $w \in \Out(v)$ 
			is randomly chosen follow the existing  policy for $w$.
			In both cases, the claim follows from the inductive assumption on
			$w$.
		
		    2) We next show that there is no policy for $\Seq(\T_{\ell_v -1})$.
			If $v \in V_1$ we have that the current vertex is not in the set $T_{\ell_{v}-1}$ and 
			no successor $w$ has a policy $\sigma$ with
			$\Pr_w^\sigma(\Seq(\T_{\ell_v -1})) = 1$ as by the inductive
			assumption $T_{\ell_v-1}$ cannot be reached a.s. from any successor of
			$v$ and thus there is also no policy $\sigma$ for $\Pr_v^\sigma(\Seq(\T_{\ell_v -1})) = 1$.
			Similar for $v \in V_R$ we have that the current vertex is not in the set $T_{\ell_{v}-1}$ and 
			there is at least one 
			successor $w$ where there is no policy $\sigma$ with  $\Pr_w^\sigma(\Seq(\T_{\ell_v -1})) = 1$
			and thus there is also no policy $\sigma$ with
			$\Pr_v^\sigma(\Seq(\T_{\ell_v -1})) = 1$ as there is a non-zero
			chance that a vertex $w$ is picked that, by the inductive
			assumption, cannot reach a node in $T_{l_v-1}$.

		\item \emph{Case $Q  = \emptyset$}: 
			As shown in the proof of Lemma~\ref{lem:seq_mdps:invariants1}(4,5) $best_w$ is not $\nill$
			for all $w \in S$ that have an edge to vertices in $V \setminus S$ and
			there is at least one vertex in $V_R \cap S$ that has an edge to $V \setminus S$.
			That is, the operator in Line~\ref{alg:seq_mdps:argmax} returns an
			argument $v$, where $best_v \not= \nill$ by the choice of $v$, and thus $\ell_v$ can be computed.
			Let $\best_{max} = max_{v \in V_R \cap S}\ best_v$.

			1) As we have no MEC (in $S$), 
			there is a policy $\sigma$, so that the play almost surely leaves 
			$S$ by using one of the outgoing edges of a random node: 
			The policy $\sigma$ can be arbitrary, except that for a player-1
			vertex $x \in S$ with an edge $(x,y)$ where $y \in V\setminus S$ we choose
			$\sigma(x) \in S$ (which must exist as $x$ would be in $Q$
			otherwise). As there are no MECs (in $S$) the
			policy $\sigma_1$ will eventually go to $V \setminus S$ using a
			random node. This implies that from 
			each vertex in $S$ player~1 has a policy to reach a vertex in
			$V \setminus S$ coming from a random vertex. By inductive assumption
			each successor of such random vertex has a policy 
			to satisfy $\Seq(\T_{\best_{max}})$. Thus it follows that
			from each vertex in $S$ player~1 has a policy to satisfy
			$\Seq(\T_{\best_{max}})$.
			Now consider the random vertex $v$ that was chosen by the 
			algorithm as $\argmax_{v \in V_R \cap S}\ best_v$. 
			By the above all successors have a policy to satisfy $\Seq(\T_{\best_{max}})$
			almost-surely. 
			Now as $L_v$ contains $\ell_v, \dots, best_v-1$ but does not contain $\ell_{v}-1$
			we obtain a policy $\sigma$ with $\Pr_v^\sigma(\Seq(\T_{\ell_v})) = 1$. 

			2) By the choice of $v$ there is also a successor (that is chosen with non-zero probability)
			that, by assumption, has no policy for $\Seq(\T_{\best_{max}-1})$ and,
			moreover, $L_v$ does not contain $\ell_{v}-1$.
			Thus, when starting in $v$ each policy will fail to satisfy $\Seq(\T_{\best_{max}-1})$ with non-zero probability, i.e., there is no policy $\sigma$ for $\Pr_v^\sigma(\Seq(\T_{\ell_v -1})) = 1$.
	\end{itemize}		
\end{proof}

\begin{prop}[Running Time]
  Algorithm~\ref{alg:seq_mdps} runs in $O(m \log n + \sum_{i=0}^k |T_i|)$ time.
\end{prop}
\begin{proof}
  Initializing the algorithm takes $O(m + \sum_{i=0}^k |T_i|)$ time. 
  This is due to the fact that we calculate $L_v$ in $O(n+m+k)$ time at~Line 1. 
  The other initialization steps take only $O(m)$ time(lines 2-6). 
  Now consider the while loop.
  Every vertex $v \in V$ is processed once. 
  The costly operations are the call of the \texttt{ProcessVertex} function and the evaluation of 
  the $\argmax$ function.
  Evaluating \texttt{ProcessVertex(v)} takes time linear in the number of
  incoming edges of $v$ plus $|L_v|$.
  Summing up over all vertices we obtain a $O(m + \sum_{i=0}^k |T_i|)$ bound.
  To compute $\argmax$ efficiently we have to maintain a priority queue containing all not yet finished random vertices.
  As we have $O(m)$ updates this costs only $O(m \log n)$ for one of the standard
  implementations of priority queues. 
  Summing up this yields a $O(m \log n + \sum_{i=0}^k |T_i|)$ running time for Algorithm~\ref{alg:seq_mdps} .
\end{proof}

By considering also the time $\mectime$ for the MEC decomposition we obtain the
desired bound and the following theorem.

\begin{theorem}\label{thm:seq:mdp:upper}
Given an MDP $P$, a starting vertex $s$ and a tuple of targets $\T = (T_1,\dots, T_k)$, 
we can calculate whether there is a player-1 policy $\sigma_1$ at $s$ for the
objective $\Seq(\T)$ in $O(\mectime + m \log n + \sum_{i=0}^k |T_i|)$ time.
\end{theorem}

\smallskip\noindent\emph{Planning in Graphs.}

\begin{algorithm}
\KwIn{DAG $D = (V,E)$ and target sets $T= \{T_1,\dots, T_k\}$}
$L_v \gets \{ i \mid v \in T_i : i = 1 \dots k \} \ \forall v \in V$\;
$A[i] \gets 0$ for $1 \leq i \leq k$\;
$\mcount_v \gets |\Out(v)|$ for $v \in V$\;
$\best_v \gets \nill$, $\ell_v \gets \nill$ for all $v \in V$\;
$\best_v \gets k+1$ for $v \in V$ with $\Out(v) = \emptyset$\;
$S \gets V$, Queue $Q \gets \{ v \in V \mid \Out(v) = \emptyset \}$\;

\While{$S \neq \emptyset$}{\label{alg:seqt_lprop:while}
	$v = Q.pop()$\;
	ProcessVertex(v)\;
}
\Return $\{v \in V \mid \ell_v = 1\}$\;
 \BlankLine
\Function{ProcessVertex(Vertex v)}{
	\lFor{$i \in L_v$}{ $A[i] \gets 1$ }
	$\ell_v \gets best_v$\;\label{alg:seqt_lprop:l_v}
	\While{$A[\ell_v -1] = 1 \land \ell_v >
		1$}{\label{alg:seqt_lprop:while2}
			$\ell_v \gets \ell_v-1$\;
		}
	\lFor{$i \in L_v$}{ $A[i] \gets 0$ }
	$S \gets S \setminus \{v \}$\;\label{alg:seqt_lprop:updateS}
	\For{$w \in \In(v)$}{
		$\best_w \gets
			\min(\best_w,\ell_v)$\;\label{alg:seqt_lprop:updatebest}
		$\mcount_w \gets
			\mcount_w-1$\;\label{alg:seqt_lprop:decrcounter}
		\If{$\mcount_w = 0 \land w \in S$}{
			Q.push($w$)\;\label{alg:seqt_lprop:pushr}
		}
	}
}
\caption{Backward Label Propagation
	Algorithm}\label{alg:seqt_lprop}
	\end{algorithm}

The algorithm for graphs works identically to the algorithm for MDPs but it
does not need the priority queue. This is due to the fact that $Q$ is always non-empty and the MEC decomposition reduces to computing SCCs.
We thus obtain a running time of $O(m + \sum^n_{i=1} |T_i|)$.
The resulting Algorithm for graphs is given as Algorithm~\ref{alg:seqt_lprop}.

\begin{prop}[Correctness]~\label{lem:lprop_correct} Given a DAG $D = (V,E)$ and
a sequential reachability objective $\Seq(\T)$ with target sets $\T = \{T_1, \dots,
T_k\}$, Algorithm~\ref{alg:seqt_lprop} determines whether a start vertex $s$ has
a path for the objective $\Seq(\T)$.
\end{prop}

\begin{observation}\label{obs:graph_seqt}
The input graph has a vertex $v$ with $\Out(v) = \emptyset$ and thus $Q$ is
non-empty after the initialization.
\end{observation}

\begin{proof}
Note that there is always a vertex $v \in V$ where $\Out(v) = \emptyset$
because we assumed that $D$ is a DAG. 
\end{proof}

The invariants below state that (a) the variables have the intended meaning, (b)
$Q$ contains all the unprocessed vertices whose successors are already
processed and (c) that the queue contains vertices as long as $S$ is not empty.

\begin{lemma}\label{lem:seqt_graph_inv1}
The following statements are invariants of the while loop at
Line~\ref{alg:seqt_lprop:while}.

\begin{enumerate}
	\item $\mcount_v = |\Out(v) \cap S|$
	\item $v \in Q$ iff. $v \in S$ and all $\Out(v) \cap S = \emptyset$.
	\item If $S$ is not empty then the queue $Q$ is not empty.
	\item $\best_v = k+1$ for all $v \in V$ with $\Out(v) = \emptyset$.
	\item $\best_v = \min_{w \in \Out(v)\setminus S} \ell_w$, for all $v \in V$
	with $\Out(v) \neq \emptyset$. In particular $\ell_w \neq \nill$ for all $w
	\in V \setminus S$.
\end{enumerate}
\end{lemma}

\begin{proof}
\begin{enumerate}
    \item The counters $\mcount_v$ are initialized as  $|Out(v)|$ and
		  $S$ is initialized as $V$. Thus the claim holds when first entering the while loop.
		  
		  Assume the claim holds at the beginning of the iteration where vertex $v$ is processed. 
		  The set $S$ is only changed in Line~\ref{alg:seqt_lprop:updateS}.
		  There $v$ is removed from the set. 
		  The counters are only changed in Line~\ref{alg:seqt_lprop:decrcounter}:
		  All counters of vertices $w$ with $v \in \Out(w)$ are decreased by one
		  (notice that $v \in \Out(w)$ iff $w \in \In(v)$).
		  Consequently $\mcount_v= |Out(v) \cap S|$ also after this iteration of the loop
		  and the claim follows.

	\item In the initial phase $S$ is set to $V$ and $Q$ is set to $\{v \in V \mid \Out(v)=\emptyset\}$. 
		  Thus the claim holds when first entering the while loop.
		  
		  Assume the claim holds at the beginning of the iteration where vertex $v$ is processed. 
		  The set $S$ is only changed in Line~\ref{alg:seqt_lprop:updateS}
		  where $v$ is removed.
		  
		  First consider a vertex $w \in Q\setminus \{v\}$. 
		  As $w$ is not removed from the set $S$ and no vertex is added to $S$
		  the claim is still true for $w$.
		  Now consider a vertex $w$ that might be added during the iteration of the loop.
		  This can only happen in Line~\ref{alg:seqt_lprop:pushr} and the if conditions 
		  ensure that $w \in S$ and $\Out(v) \cap S = \emptyset$ (by the previous invariant) and thus 
		  the claim also holds for the newly added vertices.

	\item Due to Observation~\ref{obs:graph_seqt} the claim holds when first
	entering the while loop.

	Assume the claim holds at the beginning of the iteration, where vertex $v$
	is processed. The vertex $v$ is removed from $S$ in
	Line~\ref{alg:seqt_lprop:updateS} and if the set $S$ is empty now, the claim
	follows trivially. On the other hand, if $S$ is non-empty and $Q$ is also
	non-empty the claim follows again. In the third case $S$ is non-empty and $Q$
	is empty. Assume for contradiction that no vertex is added at
	line~\ref{alg:seqt_lprop:pushr}. By invariant (2), every vertex $v \in S$ has a
	successor in $S$ as otherwise, $v$ would be in $Q$. That implies that there
	exists an SCC which is a contradiction because we assumed that $D$ is a DAG.

	\item For $v \in V$ with $\Out(v)= \emptyset$ the variables $best_v$ are initialized with $k+1$ and
		  $best_v$ is only changed in Line~\ref{alg:seqt_lprop:updatebest}
		  when a successor of the vertex is processed.
		  As $v$ has no successor, $best_v$ is not changed during the algorithm.  

	\item Note that we define the $\max$ or $\min$ over the empty set to be null.
		  As all $\ell_v$ are initialized as 
		  null and $best_v$ with $\Out(v) \not= \emptyset$ are initialized as $\nill$
		  the claim holds when the algorithm enters the loop.
		  
		  Now consider the iteration of vertex $v$ and assume the claim is true at the beginning.
		  The set $S$ is only changed in Line~\ref{alg:seqt_lprop:updateS} where $v$ is removed.
		  Let $S_{old}$ be the set at the beginning of the iteration and 
		  $S_{new} = S_{old} \setminus \{v\}$ the updated set.
		  First notice that $best_v \not= \nill$ as $v$ was in $Q$.
		  It was either initially set to $k+1$ or 
		  it was added to $Q$ when processing a vertex $w \in \Out(v)$, 
		  which would have set $best_v$.
		  Now, as $best_v \not= \nill$ the assignment in
		  Line~\ref{alg:seqt_lprop:l_v} ensures that also 
		  $\ell_v \not= \nill$.
		  For a vertex $w \in \In(v)$, the value $best_w$ is updated to $\min(\best_w,\ell_v)$
		  (Line~\ref{alg:seqt_lprop:updatebest})
		  which by assumption is equal to $\min_{x \in (\Out(w) \setminus S_{old}) \cup v} \ell_x = \min_{x \in (\Out(w) \setminus S_{new})} \ell_x$, i.e., the equation holds.
		  For vertices $w \notin \In(v)$ 
		  both $best_w$ as well as the right hand side of the equation are unchanged.
		  Hence, the claim holds also after the iteration.
\end{enumerate}
\end{proof}

From the following invariants we obtain the correctness of our algorithm.
\begin{lemma}
The following statements are invariants of the while loop at
Line~\ref{alg:seqt_lprop:while} for all $v \in V \setminus S$:

\begin{enumerate}
	\item For all $v \in V$ there exists a path $p_v \in \Seq(\T_{\ell_v})$.
	\item For all $v \in V$ there exists no path $p_v \in \Seq(\T_{\ell_v-1})$.
\end{enumerate}
	where $\T_{\ell_v} = \{T_{\ell_v},\dots, T_k\}$ or $\ell_v > k$.

\end{lemma}

\begin{proof}
	As $S$ is initialized with the set of vertices $V$ the two statements
	trivially hold after the initialization.
	
	\smallskip\noindent Now consider the iteration where vertex $v$ is processed and assume the invariants hold at the beginning of the iteration.
	We first introduce the following notation $be(v)= \min_{w \in \Out(v)} \ell_w $.
			By Lemma~\ref{lem:seqt_graph_inv1} we have  $\Out(v) \cap S = \emptyset$
			and thus also  $\ell_w \not= \nill$ for all $w \in \Out(v)$.
			Thus by Lemma~\ref{lem:seqt_graph_inv1}(5) we have $be(v)= best_v$ 
			and $\ell_v$ can be computed.
			By the while loop in Line~\ref{alg:seqt_lprop:while2} 
			we have $L_v$ contains $\ell_v, \dots, best_v-1$ but does not contain $\ell_{v}-1$.

			\begin{enumerate}
			\item We next show that there is a path $p_v$ in
			$\Seq(\T_{\ell_v})$: Let $w = be(v)$. 
			A path for vertex $w$ where $p_w \in \Seq(\T_{\ell_w})$, exists by induction hypothesis. 
			The targets $\{T_{\ell_v}, \dots, T_{\ell_{w-1}} \}$ are visited by
			starting from $v$. The path is obtained as follows: $p_v = v, p_w$,
			which proves the claim.

		    \item We next show that there is no path $p_v$ in $\Seq(\T_{\ell_v -1})$.
			The current vertex $v$ is not in the set $T_{\ell_{v}-1}$ and 
			no successor $w$ has a path $p_w$ with $p_w \in \Seq(\T_{\ell_v
			-1})$ because $\ell_{v}-1 < \ell_{w}$.
			Thus there is also no path $p_v \in \Seq(\T_{\ell_v -1})) = 1$ which
			concludes the proof.
			\end{enumerate}
\end{proof}

\begin{prop}[Running Time]~\label{lem:lprop_rtime}
Algorithm~\ref{alg:seqt_lprop} has running time  $O( m + \sum^n_{i=1} |T_i|)$.
\end{prop}

\begin{proof}
The initialization of $A$ takes $O(k)$ time. Initializing the sets $L_v$ costs 
$\sum^n_{i=1} |T_i|$ time.
We process every vertex $v \in V$ with the function $\mathit{ProcessVertex (Vertex \ v)}$
at line~9 because we assume the input graph $D$ is a DAG. 
In the function, all incoming edges of $v$ are processed once (line~15).
When processing an edge, we do constant work in lines~(18-22). 
The vertices have total work $O(\sum^n_{i=1} |T_i|)$ to do: 
Each of them goes through the list of their label three times (lines 12,14-15,16). This
yields a total running time of $O(m + \sum^n_{i=1} |T_i|)$.
\end{proof}

\begin{theorem}\label{thm:seq:graphs:upper}
Given a graph $G = (V,E)$, a starting vertex $s$ and a tuple of targets $\T = (T_1,\dots, T_k)$, 
we can calculate whether there is a player-1 policy $\sigma_1$ at a start vertex
$s$ for the objective $\Seq(\T)$ in $O(m + \sum_{i=1}^k |T_i|)$ time.
\end{theorem}

\smallskip\noindent\emph{Planning in Games.}
For game graphs with the tuple $\T = (T_1, \dots T_k)$ and starting vertex $s$, the basic algorithm
performs $k$ player-1 attractor computations, starting with computing the attractor $S_{k}=\attr_1(T_k)$ of $T_k$,
then computing $S_{\ell}=\attr_1(S_{\ell+1} \cap T_{\ell})$ for $1 \leq \ell <k$,
and finally returning $S_1$. This gives an $O(k\cdot m)$-time algorithm.

\subsection{Conditional Lower Bounds}
We present CLBs for game graphs based on the conjectures STC, SETH and OVC,
which establish the CLBs for the third row of Table~\ref{tab:complexity}.

\smallskip\noindent\emph{Sparse Game Graphs.}
For sparse game graphs, we present conditional lower bounds based on OVC.
The reduction is an extension of Reduction~\ref{red:query_mdp_ov}, where we 
(a) produce player-2 vertices instead of random vertices and 
(b) also every vertex of $S_2$ has an edge back to $s$.
The reduction is illustrated in Figure~\ref{fig:seq_games_ov}.

\begin{reduction}\label{red:seq_games_ov}
Given two sets $S_1, S_2$ of $d$-dimensional vectors, we build the following
game graph $\Gamma$.

\begin{compactitem}
	\item The vertices $V$ of the game graph are given by a start vertex $s$, 
	sets of vertices $S_1$ and $S_2$ representing the sets of vectors 
	and vertices $\C = \{c_i \mid 1 \leq i \leq d\}$ representing the coordinates of the vectors in the OVC instance. 
	\item The edges $E$ of $\Gamma$ are defined as follows: the start vertex $s$ has an edge to every vertex of
	$S_1$ and every vertex of $S_2$ has an edge back to $s$; 
	furthermore for each $x_i \in S_1$ there is an edge to $c_j \in C$ iff 
	$x_i[j] = 1$ and for each $y_i \in S_2$ there is an edge from $c_j \in S_2$ to
	$y$ iff $y_i[j] = 1$.

	\item The set of vertices is partitioned into player-1 vertices $V_1  = S_1
	\cup \C \cup S_2$ and player-2 vertices $V_2  = \{s\}$.
\end{compactitem}
\end{reduction}

\begin{lemma}
Let $\Gamma$ be the game graph given by Reduction~\ref{red:seq_games_ov} with a
tuple of target sets $\T= (T_1,\dots, T_k)$ where $T_i = \{y_i\}$ for $i = 1
\dots N$.
There exist orthogonal vectors 
$x_i \in S_1$, $y_j \in S_2$ iff $s$ has no player-1 policy $\sigma_1$ to ensure
winning for the objective $\Seq(\T)$. 
\end{lemma}

\begin{proof}
Notice that the game graph $\Gamma$ is constructed in
such a way that there is no path
between $x_i$ and $y_j$ iff they are orthogonal in the OV instance.
Notice that each play starting at $s$ revisits $s$ every four steps and 
if there is no path between $x_i$ and $y_j$ then player~2 can disrupt player~1 from visiting a target $T_j$
by moving the token to $x_i$ whenever the token is in $s$.
However, if there is no such $x_i$ and $y_j$, player~2 cannot disrupt player~1
from $s$ because no matter which vertex $x_i$ player~2 chooses,
player~1 has a policy to reach the next target set. 
If $s$ has no player-1 policy $\sigma_1$ to ensure winning for the objective
$\Seq(\T)$ there must be a target player~1 cannot reach. This must be due to the
fact that there is no path between some $x_i$ and $y_j$ and player~2 always
chooses $x_i$.
\end{proof}

The number of vertices in $\Gamma$, constructed by Reduction~\ref{red:query_mdp_ov}
is $O(N)$ and the construction can be performed in $O(N \log N)$ time (recall
that $(d = \omega(\log N))$. The number of edges $m$ is $O(N \log N)$ (thus we
consider $G$ to be a sparse graph) and the number of target sets $k \in
\theta(N) = \theta(m/\log N)$. 

\begin{theorem}\label{thm:seq:games:lower:sparse}
There is no $O(m^{2-\epsilon})$ or $O((k\cdot m)^{1-\epsilon})$ algorithm (for
any $\epsilon > 0$) to check if a vertex $v$ has a winning policy for sequential reachability objectives in
game graphs under Conjecture~\ref{conj:ovc} (i.e., unless OVC and SETH fail). 
\end{theorem}

\smallskip\noindent\emph{Dense Game Graphs.}
For dense game graphs, we present a conditional lower bound based on BMM.
The reduction extends Reduction~\ref{red:query_mdp_triangle}, where we 
(a) again  produce player-2 vertices instead of random vertices and 
(b) every vertex in the fourth copy has an edge back to $s$.
The reduction is illustrated in Figure~\ref{fig:seq_triangle_games}. 

\begin{reduction}\label{red:seq_games_triangle}
Given an instance of triangle detection, i.e., a graph $G = (V,E)$, we build the
following game graph $\Gamma = (V',E',\ls V'_1, V'_2 \rs)$.
\begin{compactitem}
	\item The vertices $V'$ are given as four copies $V_1,V_2,V_3,V_4$ of $V$
	and a start vertex $s$. 
	\item The edges $E'$ are defined as follows: There
	is an edge from $s$ to every $v_{1i} \in V_1$ where $i=1\dots n$. In
	addition for $1 \leq j \leq 3$ there is an edge from $v_{ji}$ to $v_{(j+1)k}$ iff
	$(v_i,v_k) \in E$. Furthermore there are edges from every $v_{4i} \in V_4$ to
	the start vertex $s$.

	\item The set of vertices $V'$ is partitioned into player-1 vertices $V'_1 =
	\emptyset$ and player-2 vertices $V'_2 = \{s\} \cup V_1 \cup V_2 \cup V_3 \cup
	V_4$.
\end{compactitem}
\end{reduction}

\begin{lemma}
Let $\Gamma'$ be the game graphs given by
Reduction~\ref{red:seq_games_triangle} with
a tuple of target sets $\T = (T_1, T_2, \dots, T_k)$ where $T_i = V_1 \setminus\{v_{1i}\} \cup
V_4\setminus \{v_{4i}\}$ for $i = 1 \dots k$.
A graph $G$ has a triangle iff there is no policy $\sigma_1$ to ensure winning for the objective 
$\Reach(\T)$ from start vertex s.
\end{lemma}

\begin{proof}
For the correctness of the reduction notice that there is a triangle in the graph $G$ iff 
there is a path from some vertex $v_{1i}$ in the first copy of $G$ to the same
vertex in the fourth copy of $G$, $v_{4i}$ in $P$. 
Player~2 then has a policy to always visit only 
$v_{1i}$ from the first copy and only $v_{4i}$ from the fourth copy 
which prevents player~1 from visiting target $T_i$.
\end{proof}

The size and the construction time of graph $\Gamma$, given by
Reduction~\ref{red:query_mdp_triangle}, are linear in the size of the
original graph $G$ and we have $k = \theta(n)$ target sets. 

\begin{theorem}\label{thm:seq:games:lower:dense}
There is no combinatorial $O(n^{3-\epsilon})$ or $O((k\cdot n^2)^{1-\epsilon})$
algorithm (for any $\epsilon > 0$) to check if a vertex $v$ has a winning policy
for sequential reachability objectives in game graphs under
Conjecture~\ref{conjecture:stc} 
(i.e., unless STC and BMM fail).
The bounds hold for dense game graphs with $m = \theta(n^2)$.
\end{theorem}

\section{Discussion and Conclusion}
In this work, we study several natural planning problems in graphs, MDPs, and game graphs, 
which are basic algorithmic problems in artificial intelligence. 
Our main contributions are a sub-quadratic algorithm for sequential target in
MDPs, and quadratic conditional lower bounds.
Note that graphs are a special case of both MDPs and game graphs, and the 
algorithmic problems are simplest for graphs, and in all cases, we have
linear-time upper bounds. 
The key highlight of our results is an interesting separation of MDPs and game graphs:
for basic target reachability, MDPs are harder than game graphs;
for the coverage problem, both MDPs and game graphs are hard 
(quadratic CLBs);
for sequential target reachability, game graphs are harder than MDPs.

\begin{remark}
Note that in Table~\ref{tab:complexity} in the upper bounds for MDPs 
(second column) the term $m \cdot n^{2/3}$ appears consistently, 
which is the current best-known bound for the MEC-decomposition problem. 
For all the upper bound results, any improvement for the MEC-decomposition
bound also carries over and improves the $m \cdot n^{2/3}$ term  
in all entries of Table~\ref{tab:complexity}.
Quite interestingly, for the coverage problem the CLB shows 
that the $k\cdot m$ term, which is present alongside the MEC-decomposition term,
cannot be improved (this gives quadratic CLB), 
whereas for the sequential target problem, we present a sub-quadratic upper 
bound for MDPs. 
\end{remark}

In this work, we clarified the algorithmic landscape of basic planning problems with CLBs and better algorithms.
An interesting direction of future work would be to consider CLBs for 
other polynomial-time problems in planning and AI in general. For MDPs with
sequential targets, we establish sub-quadratic upper bounds, and hence the
techniques of the paper that establish quadratic CLBs are not applicable.
Other CLB techniques for this problem are an interesting topic to
investigate as future work.

\section*{Acknowledgments}
The authors are grateful to the anonymous referees 
for their valuable comments and suggestions to improve the presentation of the paper. 

A. S. is fully  supported by the Vienna Science
and Technology Fund (WWTF) through project ICT15-003, the other authors
are partially supported by that grant.
K.C. is also supported by the Austrian Science Fund (FWF) NFN
Grant No S11407-N23 (RiSE/SHiNE) and an ERC Starting
grant (279307: Graph Games). For M.H the research leading to these results has received funding from the
European Research Council under the European Union's Seventh Framework Programme (FP/2007-2013) / ERC Grant
Agreement no. 340506.

\bibliography{mpg}
\bibliographystyle{aaai}

\end{document}